\documentclass[11pt]{article}
\usepackage{algorithm}
\usepackage[indLines = true, noEnd = true, commentColor = blue]{algpseudocodex}
\algrenewcommand\alglinenumber[1]{\tiny #1:}
\usepackage{amsmath,amsthm,amssymb}
\usepackage[a4paper, margin=1in]{geometry}
\usepackage{cite}
\usepackage{amsfonts}
\usepackage[mathscr,mathcal]{euscript}
\usepackage{times}
\usepackage{tikz}
\usetikzlibrary{positioning,calc,arrows.meta}
\usepackage{xurl}
\usepackage{csquotes}
\usetikzlibrary{patterns}
\usepackage{graphicx}
\usepackage{booktabs}
\usepackage{beramono}
\usepackage{inconsolata}    
\theoremstyle{plain}
\newtheorem{theorem}{Theorem}
\newtheorem{corollary}{Corollary}[theorem]
\theoremstyle{definition}
\newtheorem{definition}{Definition}
\theoremstyle{remark}
\newtheorem{remark}{Remark}
\usepackage{caption}
\makeatletter
\renewcommand{\paragraph}{%
  \@startsection{paragraph}{4}%
  {\z@}{2.00ex \@plus 1ex \@minus .2ex}{-1em}%
  {\normalfont\normalsize\bfseries}%
}
\let\oldbibliography\thebibliography
\renewcommand{\thebibliography}[1]{%
  \oldbibliography{#1}%
  \setlength{\itemsep}{-1pt}
}
\providecommand{\keywords}[1]
{
  \small	
  \textbf{\textit{Keywords:--}} #1
}
\makeatother
\title{A Guess and Determine Attack on the Elliptic Curve Discrete Logarithm Problem}
\author{Ayan Mahalanobis\thanks{IISER Pune, Pune, INDIA; Email: ayan.mahalanobis@gmail.com}}
\date{}
\AddToHook{env/abstract/before}{\vspace{-1.5em}}
\begin{document}
\maketitle
\begin{abstract}
This paper is a continuation of our earlier work, in which we described a Las Vegas algorithm to solve the elliptic curve discrete logarithm problem. The Las Vegas algorithm reduces the elliptic curve discrete logarithm problem to finding a zero minor in a matrix. Using the intersection poset of a hyperplane arrangement, we develop an algorithm to find a zero minor in a rectangular matrix. Our methods are elementary. We discuss the complexity and success probability of our algorithm and provide implementation details. Finding a zero minor in a matrix is also of independent interest.
\end{abstract}
\section{Introduction} This builds on our earlier work and develops a new attack on the \textbf{elliptic curve discrete logarithm problem}. We will refer to this problem as the \emph{discrete logarithm problem} in this paper. We expect our reader to be familiar with our earlier work~\cite{first,second,em}, which will be referred to often, and to be comfortable with the fact that a matrix $\mathcal{M}$, whose left-kernel is $\mathcal{K}$, can be computed from points $\mathrm{P}$ and $\mathrm{Q}$ of a non-singular elliptic curve $\mathcal{E}$ over a field $\mathbf{F}$. Furthermore, the dimension of $\mathcal{K}$ is $\ell$ where $\ell$ is an even integer. The attack we develop does not depend on the field $\mathbf{F}$ or the non-singular elliptic curve $\mathcal{E}$ chosen. 

The discrete logarithm problem is important in public key cryptography. Though it is not quantum secure, it is still used in many cryptographic protocols and signature schemes. One prominent example is the \emph{Transport Layer Security}~\cite{rfc}, in particular \texttt{secp256r1}. The signature scheme based on the discrete logarithm problem is also used in cryptocurrencies such as \emph{Bitcoin}. The discrete logarithm problem has been under constant attack since its introduction~\cite{koblitz1987, miller1985}. We will not review any of these attacks, but will focus narrowly on the attack that we developed. For a general introduction to elliptic curve cryptography, we recommend Hoffstein et al.~\cite{hoffstein2014introduction}, and for attacks Galbraith and Gaudry~\cite{gailbraith}. 

This new attack falls under the paradigm of a \emph{guess and determine} algorithm. The job at hand is to find a zero minor in a rectangular matrix. For this, we will guess a part of a zero minor and then determine the rest of it. To develop this attack, we used ideas from \emph{hyperplane arrangements}, especially \emph{central} hyperplane arrangements and also from Grassmann coordinates. Hyperplane arrangements are interesting in their own right and is an important class of geometric lattices. Their use in this algorithm is interesting, giving rise to interesting theorems and questions. Simply put, the algorithm we present is a marriage of Grassmann coordinates and hyperplane arrangements. There is an eerie similarity between our work and matrix completion which is so popular these days because of its application in \emph{machine learning}. 

As Neal Koblitz quotes~\cite{koblitz} \enquote{...in the real world if your cryptography fails, you lose
a million dollars or your secret agent gets killed}. Thus, we should pay attention to new attacks. However, we mention upfront, our attack poses \emph{no security risk} to elliptic curve cryptography at this time, see Tables~\ref{tab:pollard-probability}\, \&~\ref{prob-prod}.

\subsection{Contributions in this paper} The main contribution of this paper is a guess and determine attack on the discrete logarithm problem. In this attack, we are looking for a maximal zero-minor in a rectangular matrix $\mathcal{K}$. We showed earlier that a zero-minor solves the discrete logarithm problem (see~\cite[Section 3]{em} or ~\cite[Section 3]{second}). 

In this paper, we guess a part of the minor by choosing a few of its columns arbitrarily and then determine the other columns that make it a zero minor. There are two steps in this attack. The first one uses the guess in a recursive step and reduces the size of the matrix by half. The other step identifies a central hyperplane arrangement(see Section~\ref{hyperplane}). These hyperplanes are the left-kernels of the columns of the reduced matrix.
For this determine step, we use a positive integer $\mathtt{d}$, the defect. For a chosen $\mathcal{E}$, $\mathtt{d}$ is fixed, and the worst case complexity of the determine step of the attack is in polynomial time. 

It is easy to show that our earlier algorithm using almost principal minors~\cite[Section 5]{em} can be described in this current format, where the guesses are almost principal minors. Moreover, with this new determine step, it would be interesting to see how our earlier work fares. However, we are not repeating our earlier experiment. Instead, we focus on random guesses in the recursive step. One fundamental difference between this work and our earlier work is that we were able to exploit the anti-diagonal format of the kernel matrix $\mathcal{K}$. This gave rise to an efficient algorithm. The idea behind the defect is similar to the deviations from our earlier work.

This project started with an attempt to develop a polynomial time attack algorithm. There are two parameters in the algorithm, one is the prime $\mathtt{p}$, which is the order of the group of the discrete logarithm problem, and the other is the positive integer $\mathtt{d}$ called the defect. Here $\mathtt{d}$ depends on $\mathtt{p}$ through the probability, see Table~\ref{probtable}. For a fixed $\mathtt{d}$, our algorithm is polynomial in $\log_2{\mathtt{p}}$. However, when we take the probability into account and select the $\mathtt{d}$ appropriately, the algorithm is no better than the exhaustive search, see Table~\ref{prob-prod}. 
\paragraph{Notation} We use the following notations: $\mathrm{H}$ for hyperplanes or their basis matrix, which is the left-kernel of some column in a matrix. The matrix $\mathcal{K}$ will always be the left-kernel of $\mathcal{M}$ or its basis-matrix and is of size $\ell\times 2\ell$. We also assume that $\ell$ is an even integer throughout this paper and that $2\ell^\prime = \ell$. Furthermore, $\mathtt{d}$ is a positive integer and stands for defect. For a matrix $\mathcal{A}$, $\mathcal{A}[\mathtt{a}]$ for an array $\mathtt{a}$, whose length is less than the number of columns of $\mathcal{A}$, is the submatrix of $\mathcal{A}$, formed by taking those columns of $\mathcal{A}$ whose indices belong to $\mathtt{a}$. The index-origin for matrices, vectors, and arrays is 1, i.e., the first element is indexed by $1$, and vectors are written as rows in a matrix. We use RREF for row-reduced echelon form of a matrix. We use positive-integer interval, in which, $[n]$ stands for the set $\{1,2,\ldots,n\}$.
\section{An overview of our earlier work}
We start by defining the discrete logarithm problem, solving which is the purpose of this paper.
\begin{definition}[The discrete logarithm problem]
Let $\mathcal{E}$ be the group of rational points of an elliptic curve $\mathcal{E}$ over the field $\mathbf{F}$. Without loss of generality, we assume that the group $\mathcal{E}$ is of prime order $\mathtt{p}$ and is generated by $\mathrm{P}$. The discrete logarithm problem is: given $\mathrm{P}$ and $\mathrm{Q}$ in $\mathcal{E}$, find the integer $\mathrm{m}$, where $0<\mathrm{m}<\mathtt{p}$ and
$\mathrm{m}\mathrm{P} = \mathrm{Q}$.  
\end{definition}
Our attack can be summarized in one sentence: \emph{find a zero minor and solve the discrete logarithm problem}. And, this paper is in search for that zero minor. To understand our matrices, we first construct a matrix $\mathcal{M}$ from $2\ell$ distinct points on $\mathcal{E}$ over a field $\mathbf{F}$~\cite[Algorithm 1]{em}. Here $\ell = 3n^\prime$ is the number of points a homogeneous plane projective curve of degree $n^\prime$ is expected to intersect $\mathcal{E}$. The integer $n^\prime$ is chosen to be approximately equal to $\log_2(\mathtt{p})$, where $\mathtt{p}$ is a prime. We assume that the size of the group of the discrete logarithm problem is also $\mathtt{p}$. Henceforth, $\mathtt{p}$ will denote a prime and also the size of the group for the discrete logarithm problem.
Then we compute the left-kernel $\mathcal{K}$ of $\mathcal{M}$, which is of dimension $\ell$. We will use $\mathcal{K}$ for both the subspace and its basis matrix, this will be true for all other subspaces and their basis matrices as well. Both matrices $\mathcal{M}$ and $\mathcal{K}$ are defined over the field $\mathbf{F}$.

When we think of $\mathcal{K}$ as a basis matrix, it has size $\ell\times 2\ell$, where $\ell=\textsc{O}\left(\log_2{\mathtt{p}}\right)$. Since it is a basis-matrix, we can row-reduce it to a simpler form while keeping the subspace it spans fixed. The form that we adopted earlier is the anti-diagonal format, which means, by row reduction we reduced the matrix $\mathcal{K}$, such that the last $\ell$ columns represent a $\ell\times\ell$ matrix with one on the anti-diagonal~\cite[Equation (3)]{em} and zero everywhere else. We will call this the \textbf{anti-diagonal format}.
After reduction, we call the first $\ell$ columns the dense part. If there is a zero entry in the dense part, we have solved the discrete logarithm problem (see Problem L,~\cite{first}). So, we will safely assume in this paper that the dense part has no zero entries.

A zero minor in the dense part corresponds to a zero \emph{maximal-minor} in $\mathcal{K}$~\cite[Theorem 3.1]{em}. A maximal-minor is a minor of size $\ell\times\ell$ in $\mathcal{K}$. In this paper, we will mostly work with maximal-minors; thus, when we use a minor, \textbf{we mean a maximal-minor} unless explicitly stated otherwise.

Our Las Vegas attack proceeds as follows: we choose distinct random integers $n_i$ and $-n_j$ and create distinct points $P_i = n_i\mathrm{P}$ and $Q_j = -n_j\mathrm{Q}$ on $\mathcal{E}$. The complement of the indices of a zero minor in $\mathcal{K}$ gives rise to a linear relation of the form: 
\begin{equation}\label{eqn1}
\sum_{i}n_i-\mathrm{m}\sum_{j}n_j = 0 \mod p.
\end{equation} 
Since we know $n_i$ and $n_j$, we can solve for $\mathrm{m}$.

In this paper, we take $\ell$ many $n_i$ and $-n_j$ points, totalling to $2\ell$ points. Then we order the rows in $\mathcal{M}$, so that $P_i$ are on the top of the matrix, where the rows are indexed by $n_i$ and follow the same ordering in the set of $n_i$. In other words, $n_1$ creates $P_1$ which is the first row of $\mathcal{M}$. This is followed by $n_2$ and so forth. Once $n_i$ are exhausted,
it is followed by $-n_j$ the same way as before, and which acts on $\mathrm{Q}$, with the same ordering. Thus, the first $\ell$ rows of $\mathcal{M}$ are a scalar multiples of 
$\mathrm{P}$ followed by $\ell$ rows that are scalar multiples of $\mathrm{Q}$. Then the first $\ell$ columns of 
$\mathcal{K}$ are indexed by scalar multiples of $\mathrm{P}$ and the next $\ell$ columns are indexed by scalar multiples of $\mathrm{Q}$. In the guess and determine attack developed here, we guess half the number of columns from the dense part of $\mathcal{K}$, this is the same as choosing the $\sum_{i}n_i$ in the above equation and then we determine the $\sum_{j}n_j$, such that, the above equation is satisfied and that solves for $\mathrm{m}$, the discrete logarithm. 

Our current experiments show that this way of dividing the problem is successful. To determine columns from the sparse part, we use ideas from central hyperplane arrangements. This has two advantages. First, it reduces the complexity to polynomial time by introducing \emph{defect}. The second, and an important advantage is the introduction of signatures. The signatures reduce the problem of finding a zero minor to a much smaller matrix. However, in it lies a bottleneck of our computation. We need to do many kernel computations, though the matrices are all of the size $(\mathtt{d}-1)\times\mathtt{d}$, which is small, but still it takes up a lot of time. However, there are algorithmic optimizations, \emph{viz}., Gray codes~\cite{savage} and rank-one row replacement to create kernels of matrices which differ only by one row and then parallelization. All these techniques will alleviate the bottleneck caused by a large number of choices.

\section{Finding a mate is finding a zero minor}\label{mate}
The next idea defines the recursion in our algorithm. The idea originates in Hodge and Pedoe~\cite[Chapter VII]{HP}. In this, they describe the geometry for a Grassmann coordinate to be zero. Recall that a Grassmann coordinate is zero if and only if a maximal-minor is zero. The condition they describe is: suppose $S_k$ is a $k$-dimensional subspace of a $n$-dimensional vector space $\mathcal{V}$. Then the basis matrix of $S_k$ is a $k\times n$ matrix. 
The ordered vector of all minors of this matrix is the Grassmann coordinate of the subspace. Clearly, a Grassmann coordinate is not unique as many different bases may generate the same subspace, however they are unique up to a scalar multiple. Thus, one can ask: when is a Grassmann coordinate zero? Let $e_i=(0,0,\ldots,1,\ldots 0)$ be the standard basis of $\mathcal{V}$ with $1$ in the $i$th place and $0$ everywhere else. Let $\textsc{E}$ be the subspace generated by $k$ standard basis elements. Then $S_k$ has a zero Grassmann coordinate if and only if $S_k\cap\text{E}\neq 0$ for some $\textsc{E}$ and the zero Grassmann coordinates are given by the basis vectors that generate $\textsc{E}$. Grassmann coordinates are also known as Pl\"{u}cker coordinates, see~\cite[Chapter 10]{shafarevich}. Recall, the left-kernel $\mathcal{K}$ is in the anti-diagonal format and is of size $\ell\times 2\ell$. The first index of rows and columns is $1$.
The following theorem is a formalization of the above argument. 

\begin{theorem}
Let $S$ be the set $\{nr+1, nr+2, \ldots, nr+\ell\}$ which is the set of all column indices in the sparse part. For $0<k<\ell$, let $c_1,c_2,\ldots,c_k$ be a set of columns from the dense part of $\mathcal{K}$. 
Let $\mathcal{K}^\prime$ be the submatrix of $\mathcal{K}$ formed by these columns. Then $\mathcal{K}^\prime$ is of size $\ell\times k$. There is a  zero minor in $\mathcal{K}^\prime$, if and only if, there is a zero minor in $\mathcal{K}$. Note, $\mathcal{K}^\prime$ is a tall matrix. If the zero minor of $\mathcal{K}^\prime$ is from rows $r_1,r_2,\ldots,r_k$, and $S^\prime=S\smallsetminus\{2\ell-r_1+1,2\ell-r_2+1,\ldots,2\ell-r_k+1\}$ then the zero minor of $\mathcal{K}$ consists of columns $\{c_1,c_2,\ldots,c_k\}\cup S^\prime$. Clearly, $S^\prime$ is of size $\ell-k$ and thus the total number of columns is $\ell$.
\end{theorem}

\begin{proof}
The proof of this theorem is fairly straightforward and involves no Grassmann coordinates. First, assume that $\mathcal{K}$ has a zero minor and that it consists of both the dense and sparse parts. Then there are column operations in that minor that make a zero column in that minor. This means that, by column operations on the dense part, we will produce a vector whose zero positions will match those arising from the sparse part. This is due to the special structure of the columns in the sparse part. Thus, column operations of columns from the dense part will have a zero in a particular position, if and only if, the column with $1$ in that position from the sparse part was not included to form the minor.

Conversely, if there is a zero minor in $\mathcal{K}^\prime$ with rows $r_1,r_2,\ldots,r_k$ then there are row operations on $\mathcal{K}^\prime$ that give rise to a vector with zeros in the $r_1,r_2,\ldots,r_k$ rows. Then, if we add columns to $\mathcal{K}^\prime$ from the sparse part that have zeros in those positions, we will get a zero column in the extended minor. It is a simple exercise to see that the number of columns from the dense part and the number of columns from the sparse part add to $\ell$. This completes the proof.    
\end{proof}
\begin{remark}
In our earlier work~\cite[Theorem 3.1]{em}, we proved that there is a one-one correspondence between (non-maximal) zero minors in the dense part and maximal minors of the whole matrix. We now examine this situation using the above theorem. If there is a (non-maximal) zero minor in the dense part, create a submatrix of columns from the zero minor. The zero minor is then a maximal minor of the submatrix. The above theorem gives us the columns from the sparse part that give rise to a maximal minor in the original matrix. Thus, (non-maximal) zero minors give rise to maximal zero-minors. 

Conversely, for a maximal zero minor in the original matrix, the columns from the sparse part give rise to rows of a maximal zero minor in the submatrix of the dense part obtained from the columns in the dense part. Thus, we have rows and columns for a (non-maximal) zero-minor in the dense part. 
\end{remark}

This theorem leads us to a \enquote{finding a mate} algorithm, which is fairly straightforward to describe and defines our recursion! First select an index set $\mathtt{a}$ with index of $k (<\ell)$ columns from the dense part of $\mathcal{K}$ and form the submatrix $\mathcal{K}^\prime=\mathcal{K}[\mathtt{a}]$. Then take the transpose of $\mathcal{K}^\prime$, which is now a $k\times\ell$ matrix. Then do row operations on this transpose to convert it to the anti-diagonal format. Then take another set of columns from the dense part of this converted matrix and keep doing that, until the resultant matrix is small enough to do an exhaustive search for a zero minor. If there is a zero minor in the resultant matrix there is a zero minor in the original matrix $\mathcal{K}$. However, from our computational experience the above method does not work. In particular, if the number of steps in the recursion is too many, we will end up with a matrix with no zero minors. In the first case, $k$ choices were made, and the set of columns from the sparse part completes these choices by producing a zero minor. When we do the next step in the recursion, we are making more choices from the columns of $\mathcal{K}$. These choices that we make consecutively, reduce the possibility of finding a mate, and we end up being single. Thus, in describing the attack algorithm, we use recursion only once, reducing the size of the matrix by half and making the algorithm computationally efficient. We are not saying that higher-order recursions will never help. As the size of the matrix $\mathcal{K}$ grows with the size of the field, so will the number of zero minors in it. Then we can use two-step recursion with the same effect as we have now.
\section{Hyperplane arrangements and singular matrices}\label{hyperplane}
A square matrix $\mathcal{A}$ is singular if its determinant is zero. This paper is about a search for a square matrix $\mathcal{A}$ (maximal-minor) in the rectangular matrix $\mathcal{K}$ (left-kernel). 
In a vector space $\mathcal{V}$ of dimension $n$, any subspace of dimension $n-1$ is called a \emph{hyperplane}. In the literature, there is a distinction between linear hyperplanes and affine hyperplanes. Since all our hyperplanes are linear hyperplanes, we will use \emph{hyperplanes} for \emph{linear hyperplanes}. A finite collection of hyperplanes is called a \emph{hyperplane arrangement}; we will use the term arrangement from now on. The set of all intersections of these hyperplanes forms a poset and is called the \emph{intersection poset} of the arrangement, which is ordered by reverse inclusion and is known to form a geometric lattice. Our standard reference for arrangement is Stanley~\cite[Chapter 3]{stanley} or~\cite{stanley1}. Corresponding to a $n\times n$ matrix $\mathcal{A}$ with columns $c_1,c_2,\ldots,c_n$ we can define $n$ hyperplanes $\mathcal{K}_1,\mathcal{K}_2,\ldots,\mathcal{K}_n$ as the left-kernels of these columns, respectively. We will refer to this particular arrangement as the arrangement from the matrix $\mathcal{A}$. Implicit is the assumption that the columns are of rank one. We will assume this throughout this paper. In a straightforward way, one can extend this arrangement from a square matrix $\mathcal{A}$ to a rectangular matrix $\mathcal{K}$ where the hyperplanes are the left-kernel of its columns. However, the intersection poset has to be defined carefully. We define the poset to be the set of intersections of at most $\ell$ many hyperplanes ordered by reverse inclusion. The vertices of this poset are labelled by the subspace formed by the intersection of hyperplanes.

The matrix $\mathcal{K}$ being a rectangular basis-matrix of size $\ell\times 2\ell$ is a proper subspace of $\mathcal{V}$ with the vectors written as rows. Thus, in our context, $\mathcal{V}$ is of dimension $2\ell$ and $\mathcal{K}$ is of dimension $\ell$.
Our main interest is in the intersection of $\ell$ hyperplanes from the set of all hyperplanes $\left\{\mathcal{K}_i\right\}$ for $i$ in $\{1,2,\ldots,2\ell\}$. The intersection of $\ell$ hyperplanes is the \emph{ultimate intersection}, and the intersection of $\ell-1$ hyperplanes will be called a \emph{penultimate intersection}. One can define a subarrangement of an arrangement in an obvious way. So, instead of using all the columns of $\mathcal{K}$, we can take a subset of columns of $\mathcal{K}$ and the corresponding left-kernels will form a subarrangement. Needless to say, we are interested in subarrangements of $\ell$ hyperplanes of $\mathcal{K}$.  

We briefly review the concept of \emph{general position} for intersection posets (see,~\cite[Page 287]{stanley}). It says that a set $\left\{H_1,H_2,\ldots,H_t\right\}$ of hyperplanes is in general position, if $\dim(H_1\cap H_2\cap\cdots\cap H_t)= \ell -t$ whenever $t\leqslant\ell$. Furthermore, if $t>\ell$, the dimension of the intersection is $0$. Here $\ell$ is the size of the arrangement. Translated in our context of finding a zero minor, if we find a set of $t$ hyperplanes that are \emph{not in general position}, we have found at least one zero minor. This is because those columns are linearly dependent, and a minor containing those columns will be a zero minor. In our experiments, we have not seen this happen, except when $t=\ell$. We do not check for general position in the rest of this work or in our algorithms. However, that can be easily implemented.

A hyperplane arrangement is called \emph{central} if there is an ultimate intersection that is non-zero; we will refer to this as a non-trivial ultimate intersection. Thus, in a central arrangement for $\mathcal{K}$ there is a non-zero vector $v\in\mathcal{V}$, such that, $v\in\cap_{i\in\Upsilon} \mathcal{K}_i$ for some subset $\Upsilon$ of $[2\ell]$ of cardinality $\ell$. We call this vector $v$ a central element of the hyperplane arrangement. Next theorem is about a square matrix. 
\begin{theorem}\label{thm1}
In a hyperplane arrangement $\left\{\mathcal{K}_i\right\}_{i=1}^n$ for a $n\times n$ matrix $\mathcal{A}$ over the field $\mathbf{F}$ the following are equivalent:
\begin{description}
\item[a)] The matrix $\mathcal{A}$ is a singular matrix.
\item[b)] There is a nonzero vector $v$ of length $n$ over $\mathbf{F}$ such that $v\mathcal{A}=0$.
\item[c)] The ultimate intersection $\bigcap\limits_{i=1}^n\mathcal{K}_i\neq 0$. 
\end{description}
\end{theorem}
\begin{proof}
The proof revolves around the fact, that a singular matrix $\mathcal{A}$, the corresponding hyperplane arrangement $\left\{\mathcal{K}_i\right\}_{i=1}^n$ is central. One way to think about it, is to construct these $n$ intersections $\mathcal{K}_1,\mathcal{K}_1\cap\mathcal{K}_2,\ldots,\mathcal{K}_1\cap\mathcal\cdots\cap\mathcal{K}_n$. Then compute a matrix $\mathrm{T}$ whose first row is a non-zero element from $\mathcal{K}_1$, the second row is a non-zero element from the second intersection $\mathcal{K}_1\cap\mathcal{K}_2$ and continue selecting non-zero elements from respective intersections till the $n^\texttt{th}$ row which is a non-zero element from $\mathcal{K}_1\cap\mathcal{K}_2\cap\cdots\cap\mathcal{K}_n$. This construction follows: using the Gaussian elimination algorithm to transform $\mathcal{A}$ to an upper-triangular matrix using row operations. In other words, $\mathrm{T}\mathcal{A}$ is an upper-triangular matrix. Furthermore, when $\mathcal{A}$ is singular, the last row is zero in this reduced upper triangular matrix, and the last row of $\mathrm{T}$ is non-zero. This last row of $\mathrm{T}$ is a scalar multiple of $v$.
This proves the existence of a non-zero $v\in\mathcal{K}_1\cap\mathcal{K}_2\cdots\cap\mathcal{K}_n$. This $v$ is a central element of the arrangement. Furthermore, $v\mathcal{A} = 0$. This is because $v$ belongs to the left-kernel of all columns of $\mathcal{A}$. This proves (a) implies (b).

Similarly, (b) implies (c) follows from the fact that $v\mathcal{A}=0$ says that all hyperplane intersections are non-zero. Similarly, (c) implies (a) follows from the fact that $\mathrm{T}$ can be constructed as described above with no zero rows.
\end{proof} 
In other words, finding a zero minor is finding a central subarrangement of size $\ell$ in the intersection poset of $\mathcal{K}$. In the arrangement of $\mathcal{K}$, we have $\hat{0}=\mathcal{V}$ as the minimal element of the intersection poset, which is the empty intersection of hyperplanes. The atoms of this poset are the hyperplanes. However, this poset might not have a $\hat{1}$ -- the maximal element. However, one can introduce a $\hat{1}$ in the intersection poset of $\mathcal{K}$ as the top element.

As we saw earlier, an arrangement being central is the same as the existence of a central element $v$. This $v$ will also be the intersection of all penultimate intersections. Thus, this will make some maximal chains in the poset shorter by one unit than they would have been if there were no central element. The intersection poset of $\mathcal{K}$ looks different when $\mathcal{K}$ has a zero minor compared to when it does not. If the minor under consideration is non-singular, then its intersection poset will be a Boolean lattice with a $\hat{1}$. However, if it is singular, the top of the poset will not exist as $v$ will generate the ultimate and all penultimate intersections. Thus, finding a zero minor is to look for this anomaly and find its properties in the intersection poset of $\mathcal{K}$.

The next theorem is vital to the algorithm that we describe later. In this theorem, we define the \textbf{signature} of an intersection of a hyperplane $\mathrm{H}$ with a subspace $\mathrm{T}$. Recall, the \emph{span of a matrix is its row-span}. This means that we write the vectors of $\mathcal{V}$ as rows in a matrix.
\begin{theorem}\label{thm2}
Let $\mathrm{H}$ be a basis-matrix of a hyperplane in a vector space $\mathcal{V}$ of dimension $\mathtt{n}$. Let $\mathrm{T}$ be a basis-matrix of a subspace of dimension $\mathtt{k}$ of $\mathcal{V}$, where $0 < \mathtt{k} < n$. Then there is a vector of the form $\mathtt{s}=\left(s_1,s_2,\ldots,s_{k}\right)$, where the first non-zero element in $\mathtt{s}$ is $1$. If $\mathrm{S}$ is a basis-matrix of the left-kernel of the transpose of $\mathtt{s}$, then $\mathrm{S}\mathrm{T}$ is a basis-matrix of the intersection of the subspace spanned by $\mathrm{H}$ and the subspace spanned by $\mathrm{T}$. The vector $\mathtt{s}$ is the signature of this intersection and can be $0$.
\end{theorem}
\begin{proof}
Let the rows of $\mathrm{H}$ be $\{h_1,h_2,\ldots,h_{n-1}\}$ which are a basis of the span of $\mathrm{H}$. Then any element 
$\mathtt{h}\in\mathrm{H}$ is of the form $\alpha_{1}h_1 + \alpha_{2}h_2+\cdots+\alpha_{n-1}h_{n-1}$, where $\alpha_i\in\mathbf{F}$. Similarly, any element $\mathtt{t}$ in the span of $\mathrm{T}$ can be written as $\beta_1t_1 + \beta_2t_2 + \cdots + \beta_kt_k$ where $t_i$ are the rows of $\mathrm{T}$ for $1\leq i\leq k$, $\beta_i\in\mathbf{F}$. Then 
$\mathtt{t}$ is an arbitrary element of the span of $\mathrm{T}$. Now consider the matrix $\mathrm{U}$ which is formed by the transpose of $\mathrm{H}$ followed by the transpose of $\mathrm{T}$. Thus $\mathrm{U}$ has 
$\mathtt{n}$ rows and $\mathtt{n}-1+\mathtt{k}$ columns.

Now we compute a RREF on $\mathrm{U}$. We will end up with a matrix with the last row made of $\mathtt{n}-1$ zeros followed by the signature $\mathtt{s} = \left(s_{1},s_2,\ldots,s_{k}\right)$. A row operation ensures that the first non-zero element of $\mathtt{s}$ is $1$. 

We need to show that span$(\mathrm{S}\mathrm{T})=\textrm{span}(\mathrm{H})\bigcap\textrm{span}(\mathrm{T})$. Let $b\in\textrm{span}(\mathrm{S}\mathrm{T})$. Then clearly $b\in\textrm{span}(\mathrm{T})$, and we only need to prove that $b\in\textrm{span}(\mathrm{H})$. Since each row $b$ of the matrix $\mathrm{ST}$ is a basis, it is enough to prove when $b$ is a row of $\mathrm{ST}$. Then $b = s_{i1}t_1+s_{i2}t_2+\cdots+s_{ik}t_k$ where $\left(s_{i1},s_{i2},\ldots,s_{ik}\right)$ is the $i^\texttt{th}$ row of $\mathrm{S}$. Recall $\mathrm{U}$ and the RREF and the $\left(s_{i1},s_{i2},\ldots,s_{ik}\right)$ is in the kernel of the signature. This says, that there is a solution to the equation 
$\mathrm{H}^{\textrm{T}}x = b$ and thus $b\in\textrm{span}(\mathrm{H})$.

Conversely, assume that $b$ is in the intersection. Then $\mathrm{H}^{\textrm{T}}x = b$ has a solution. Furthermore, if $b = \beta_{1}t_1+\beta_{2}t_2+\cdots+\beta_{k}t_k$ where $\beta_i\in\mathbf{F}$, then $\left(\beta_1, \beta_2, \ldots, \beta_k\right)$ must belong to the kernel of the signature $\left(s_1,s_2,\ldots,s_k\right)$. For this, recall $\mathrm{U}$ and then the RREF on that. 
\end{proof}
\begin{corollary}
Using notations from the above theorem, if the signature $\mathtt{s}\neq 0$, the span of the matrix $\mathrm{ST}$ is of dimension $\mathtt{k}-1$. If $\mathtt{s} = 0$, the span of $\mathrm{T}$ is a subspace of the span of $\mathrm{H}$.
\end{corollary}
\begin{proof}
A linear combination of rows of $\mathrm{ST}$ is $\alpha_1c_1\mathrm{T} + \alpha_2c_2\mathrm{T}+\cdots+\alpha_{\mathtt{k}}c_\mathtt{k}\mathrm{T}$ where $c_1,c_2,\ldots,c_\mathtt{k}$ are the rows of $\mathrm{S}$ and $\alpha_i$ are scalars. Thus $\alpha_1c_1\mathrm{T} + \alpha_2c_2\mathrm{T}+\cdots+\alpha_{\mathtt{k}}c_\mathtt{k}\mathrm{T}= 0$ implies that $\alpha_1c_1+\alpha_2c_2+\cdots+\alpha_kc_k$ is in the left-kernel of $\mathrm{T}$. Since $\mathrm{T}$ is a basis matrix of size $\mathtt{k}\times \mathtt{n}$ the left-kernel is $0$ and the fact that $c_1,c_2,\ldots,c_{\mathtt{k}}$ are linearly independent, implies $\alpha_i = 0$ for all $i$. This proves that the row-span of $\mathrm{ST}$ is of dimension $\mathtt{k}-1$.  

If $\mathrm{s} = 0$, then $\mathrm{S}$ is the identity matrix and the assertion follows.
\end{proof}
The above theorem is useful for the algorithm we develop. In the determine step, we will start with a matrix $\mathcal{K}^\prime$ of size $\ell^\prime\times\ell$ where $\ell = 2\ell^\prime$. This matrix is in the anti-diagonal format.
From this matrix, $\ell^\prime - \mathtt{d}$ columns from the left-half is extracted, and a submatrix thus formed is called $\textsc{M}$. The subspace $\mathrm{T}$ in the above theorem is the left-kernel of $\textsc{M}$. Clearly, the matrix $\textsc{M}$ is of size $\ell^\prime\times(\ell^\prime-\mathtt{d})$ and the dimension of $\mathrm{T}$ is $\mathtt{d}$. 

In the determine part of the algorithm, we want to extend these columns of $\textsc{M}$ to form a central ultimate hyperplane arrangement in $\mathcal{K}^\prime$. To do this, we could check each hyperplane of $\mathcal{K}^\prime$ computed from columns that are not in $\textsc{M}$. But the above theorem gives us an efficient way to do that. The signature $\mathtt{s}$ in the above theorem can be easily computed from the RREF of the hyperplane and $\mathrm{T}$, and can be considered as a representative of their intersection.

Moreover, the kernel of the signature behaves well with the intersection. In other words, since the dimension of $\mathrm{T}$ is $\mathtt{d}$, $\mathrm{ST}$ will have dimension $\mathtt{d}-1$ as a subspace, where $\mathrm{S}$ is the kernel of $\mathtt{s}$ and the kernel of two signatures when multiplied to $\mathrm{T}$ will be a subspace of dimension $\mathtt{d}-2$ and the kernel of $\mathtt{d}-1$ signatures when multiplied to $\mathrm{T}$ will be a subspace of dimension $1$ -- a vector. Since the subspace $\mathrm{T}$ is fixed, instead of multiplying $\mathrm{T}$ with the kernel of signatures, we can use the signatures and their kernels as representatives of the intersection of the hyperplanes with $\mathrm{T}$. This is particularly relevant in the $\mathtt{d}-1$ case where the kernel of the signatures becomes one-dimensional. Then the intersection of $\mathtt{d}-1$ hyperplanes and $\mathrm{T}$ can be represented by a signature, which is a vector of size $\mathtt{d}$, and by making the first non-zero entry of this vector one, it is even unique.
 
\subsection{Creating the Signature-Matrix} A signature is a row-vector of length $\mathtt{d}$, which is much smaller than the size of vectors in $\mathrm{T}$ or $\mathrm{H}$ which is $\ell^\prime$. This provides a major computational advantage. Once we have the signatures, we will create the signature-matrix for the left-kernel $\mathrm{T}$ of $\textsc{M}$. In a signature-matrix, we will stack all signatures of all hyperplanes from all columns of $\mathcal{K}^\prime$, which are not in $\textsc{M}$, as rows. This matrix will be a tall but slim matrix with $\mathtt{d}$ columns and $\ell^\prime+\mathtt{d}$ rows. It is straightforward to argue that if we have a zero minor in this signature-matrix, we have a zero minor in the original matrix $\mathcal{K}$. However, in practice, we will not compute $\mathtt{d}\times\mathtt{d}$ minors of the signature matrix, which corresponds to an ultimate intersection in $\mathcal{K}^\prime$. Instead, we will compute all possible penultimate intersections, which is the kernel of matrices of size $(\mathtt{d}-1)\times\mathtt{d}$ from the signature-matrix. These kernels will be one-dimensional. We can store this generator vector in an array and check for duplicates in that array. Duplicate penultimate intersections will give rise to a zero minor in the signature-matrix. We can also use a hashtable in an obvious way by hashing the unique generator vector to speed up duplicate detection. Furthermore, we can do this check for duplicates even when creating the hashtable. Using this idea, we can stop at the first occurrence of a duplicate. This will stop the waste of resources, where we create the whole table and then check for duplicates.  

As we will see next, hyperplanes for columns in the sparse part are easy to compute, and computing their \textbf{signatures are free}! Thus, $\ell^\prime$ signatures in the signature-matrix are free to obtain and we actually need to compute only $\mathtt{d}$ of them by RREF.

Let $\mathscr{A}$ be a minor of $\mathcal{K}$. Except for one case, $\mathcal{A}$ will have columns from both the dense and the sparse part of $\mathcal{K}$. Let $c_1,c_2,\ldots,c_k$ be the columns from the sparse part of $\mathcal{K}$. Then $c_i$ is a column-vector of length $\ell$ with only one $1$ and the rest zero. The position of the one in $c_i$ is $(i,2\ell-i+1)$ in $\mathcal{K}$ for $i\in[\ell]$.
Then the left-kernel of a $c_i$ is easy to determine. It is the $\ell\times\ell$ identity matrix with one row deleted. The deleted row is the one that is the transpose of $c_i$. When we are talking about columns $c_1,c_2,\ldots,c_k$, the left-kernel of the submatrix of these columns are the $\ell\times\ell$ identity matrix with the $k$ rows deleted corresponding to each $c_i$. If $v$ belongs to this left-kernel, then $v$ is a linear sum of these rows and has $k$ zeros, exactly where each $c_i$ had a $1$ for $i$ in $\{1,2,\ldots,k\}$. This argument applies to $\mathcal{K}^\prime$ as well, when it is in anti-diagonal format. Now let us look at the columns of $\mathcal{A}$ from the dense part of $\mathcal{K}$. There will be $\ell - k$ of those. While computing signatures from these hyperplanes coming from (sparse) columns, we note that 
RREF is then just a row exchange of the transpose of $\mathrm{T}$. To see this, note that the hyperplane for the column $c_i$ where $i\in[\ell,2\ell]$ is the $\ell\times\ell$ identity matrix with the $\ell-i+1$ row removed. Then, for the RREF, we first have to transpose it and then send this row to the bottom. That will not be in RREF form but will be good enough for our purpose. Then that exchanges the last row of the transpose of $\mathrm{T}$ with the $\ell-i+1$ row. Thus, to compute the signature of the hyperplane from $c_i$ from the sparse part of $\mathcal{K}$, we just have to find the $\ell-i+1$ row of the transpose of the basis matrix of the left-kernel $\mathrm{T}$. This argument applies \emph{ad verbatim} for $\mathcal{K}^\prime$ when it is in the anti-diagonal format. 

In the algorithm below, we choose a subset $\mathtt{b}$ of size $\ell^\prime -\mathtt{d}$ from $[\ell^\prime]$.
Then we have to compute the RREF for hyperplanes that are in $[\ell^\prime]\smallsetminus(b)$ and then stack the whole basis matrix of the transpose of $\mathrm{T}$ under it. Then, when a duplicate kernel of submatrices of size $(\mathtt{d}-1)\times\mathtt{d}$ of the signature-matrix is found, we have found a zero-minor in the signature matrix. The indices of the rows of the zero-minor in the signature-matrix are converted to column indices of $\mathcal{K}^\prime$ at the post-processing stage (Algorithm~\ref{algo2}). This section also caters to the \textsc{create-signature} and \textsc{signature-mat} from $\mathcal{K}^\prime$ in Algorithm~\ref{alg1} leading to the creation of the matrix $\mathrm{A}$ which is of size $(\ell^\prime+\mathtt{d})\times\mathtt{d}$. Then the table $\mathrm{A}^\prime$ is the hashtable, which is derived from $\mathrm{A}$ and has a maximum length of $\binom{\ell+\mathtt{d}}{\mathtt{d}-1}$.

\section{An algorithm to find a zero minor}
Now we describe the main algorithm of this paper, which finds a zero minor in the rectangular matrix $\mathcal{K}$ of size $\ell\times 2\ell$. The first thing to do is apply the recursion once and reduce the size of the matrix by half. To do this, it is important that $\ell=2\ell^\prime$ is an even number. Since, the number of distinct random points on the elliptic curve selected to form the matrix $\mathcal{M}$ is our choice, this can be easily achieved. We call the reduced matrix 
$\mathcal{K}^\prime$, which is of size $\ell^\prime\times\ell$ and is in the anti-diagonal format. 

The next step of the algorithm uses ideas from hyperplane arrangements, where the hyperplanes are the left-kernel of the columns of $\mathcal{K}^\prime$. Recall that the intersection posets of these hyperplane arrangements are ordered by reverse inclusion.  We are trying to find a central element in this intersection poset. A central element corresponds to an array of size $\ell^\prime$ of columns of $\mathcal{K}^\prime$ whose hyperplanes have a non-zero intersection. We could do an exhaustive search of all possible combinations of all columns, but then there will be 
$\binom{\ell}{\ell^\prime}$ combinations, and which is of exponential growth. 

However, the above idea is not that useless. 
We switch back to $\mathcal{K}$ and show that we can develop a Pollard's rho algorithm where the random walk is on the set of penultimate intersections of $\ell -1$ hyperplanes of columns from $2\ell$ columns of $\mathcal{K}$. Note that the ultimate intersection corresponds to intersection of $\ell$ hyperplanes and a penultimate intersection corresponds to intersection of $\ell -1$ hyperplanes. Furthermore, when two penultimate intersections are the same subspace we have a non-zero ultimate intersection which leads to a zero-minor.
      
For an overview of Polard's rho we recommend Hoffstein et. al~\cite[Section 4.5]{hoffstein2014introduction}. The random walk in the set of penultimate hyperplane intersections are creating one penultimate intersection after another, which are one-dimensional subspace. Thus, they can be uniquely represented by a vector. Once there is a collision in this random walk, we have two penultimate hyperplanes with the same subspace, and the discrete logarithm problem is solved.

In the case of hyperplanes, the expected number of steps for a collision is $\sqrt{\binom{2\ell}{\ell -1}}$ and depends on the probability of success. We tabulated below the result of a small experiment, Table~\ref{tab:pollard-probability}. In which, we computed the logarithm base 2 of $\binom{2\ell}{\ell-1}$ and  we  compared it with the logarithm base 2 of the prime which we assume to be the size of the group of the elliptic curve. This way we compared Pollard's rho on penultimate intersections with the original Pollard's rho attack on an elliptic curve of same size. This comparison was made meaningful with the probability of success as described after Equation~\ref{defect_eqn}. For this experiment, corresponding to a prime $\mathtt{p}$, first we computed the integer $n^\prime$ to be the floor of $\log_2{\mathtt{p}}$. Then we reduced $n^\prime$ by dividing it with $5 + \mathtt{tt}$, where $\mathtt{tt} = 0.1\mathtt{t}$ and $\mathtt{t}\in\{0,1,\ldots,9\}$. Then $\ell=3n^\prime$ is the number of rows of $\mathcal{K}$ whose number of columns is $2\ell$. Thus, we made $\mathcal{K}$ smaller. Since, it is the probability that matters we are at liberty to change $\ell$. We did the same for all $\mathtt{tt}$.
We tabulated $\log_2\mathtt{p}$, the binomial $\log_2{\binom{2\ell}{\ell-1}}$ and the probability as a row in Table~\ref{tab:pollard-probability}, and then picked those rows that gave us a probability of $0.9$ or more. The table below shows that our Polard's rho algorithm on the set of penultimate intersections is comparable to that of the original Pollard's rho. Further research is needed for a proper implementation.
\begin{table}[htbp]
    \centering
    \caption{Selected parameters with the highest probability for values of $\log_2\binom{2\ell}{\ell-1}$ and $\log_2\mathtt{p}$. The data shows that the Pollard's algorithm on the space of penultimate hyperplane intersections is about the same as the original Pollard's rho algorithm.}
    \label{tab:pollard-probability}
    \scriptsize
    \begin{tabular}{rrr||@{\qquad}rrr}
        \toprule
        $\log_2 \mathtt{p}$ & $\log_2\binom{2\ell}{\ell-1}$ & Probability
        & $\log_2 \mathtt{p}$ & $\log_2\binom{2\ell}{\ell-1}$ & Probability \\
        \midrule
        102 & 104 & 0.99188 & 179 & 181 & 0.99939 \\
        108 & 110 & 0.99081 & 185 & 187 & 0.99931 \\
        114 & 116 & 0.98970 & 191 & 193 & 0.99924 \\
        120 & 122 & 0.98854 & 197 & 199 & 0.99915 \\
        126 & 128 & 0.98735 & 203 & 205 & 0.99906 \\
        132 & 134 & 0.98612 & 209 & 211 & 0.99897 \\
        138 & 140 & 0.98485 & 215 & 217 & 0.99887 \\
        144 & 146 & 0.98355 & 221 & 223 & 0.99877 \\
        150 & 152 & 0.98223 & 227 & 229 & 0.99866 \\
        155 & 157 & 0.99963 & 233 & 235 & 0.99854 \\
        161 & 163 & 0.99958 & 239 & 241 & 0.99842 \\
        167 & 169 & 0.99952 & 245 & 247 & 0.99830 \\
        173 & 175 & 0.99946 &     &     &         \\
        \bottomrule
    \end{tabular}
\end{table}

Now we describe the main algorithm of this paper and for this we move back to $\mathcal{K}^\prime$. For this we choose $\mathtt{d}$, a small positive integer, called the defect. For many of our experiments, we took $\mathtt{d}=4$ or $5$. However, a proper size for $\mathtt{d}$ is obtained from a probability table computed with the probability described after Equation~\ref{defect_eqn}. For an example see Table~\ref{probtable}.
Let $\mathtt{b}$ be a set of $\ell^\prime -\mathtt{d}$ columns from the dense part of $\mathcal{K}^\prime$. Let $\textsc{M}$ be the submatrix constructed from these columns of $\mathtt{b}$. We then compute the left-kernel of $\textsc{M}$ and denote it  by $\mathrm{T}$. The dimension of $\mathrm{T}$ is $\mathtt{d}$ under general position. This is one spot where we can parallelize the algorithm with different processes getting different $\mathtt{b}$. This type of parallelization is called a ridiculous parallelization, which requires no inter-process communication.

Now we need to find $\mathtt{d}$ hyperplanes, from the columns of $\mathcal{K}^\prime$ that are not in 
$\textsc{M}$, whose intersection among themselves and then with $\mathrm{T}$ is non-zero. This provides the existence of a central element, and then the columns in $\textsc{M}$ together with the columns whose hyperplanes had non-zero intersection give us a zero minor of $\mathcal{K}^\prime$, and then that can be extended to a zero minor in $\mathcal{K}$. 

\begin{figure}[htbp]\label{fig1}
\centering
\setlength{\fboxsep}{6pt}
\fbox{
\begin{tikzpicture}[
    box/.style={
        draw,
        rounded corners,
        very thick,
        minimum width=2.4cm,
        minimum height=1.5cm,
        align=center
    },
    circ/.style={
        draw,
        circle,
        very thick,
        minimum size=2.2cm,
        align=center
    },
    edge/.style={
        -{To[length=2mm,width=2mm]},
        very thick
    }
]

\node[circ] (B1) at (0,0) {$\textsc{DLP}$\\ solver};

\node[box] (B2) at (4,0) {$\mathcal{K}$ \\size $\ell\times 2\ell$};
\node[box] (B3) at (8,0) {$\mathcal{K}^\prime=\mathcal{K}[\mathtt{a}]^\textsc{T}$ \\size $\ell^\prime\times \ell$};

\node[box] (B5) at (2,-4) {$\textsc{signature-matrix}$\\[1mm]$\mathrm{A}$ of size $(\ell^\prime+\mathtt{d})\times\mathtt{d}$
\\[1mm]$\textsc{duplicate-detection}$};
\node[box] (B4) at (7,-4) {$\textsc{M}=\mathcal{K}^\prime[\mathtt{b}]$\\size $\ell^\prime\times (\ell^\prime-\mathtt{d})$};


\draw[edge] (B2) -- node[above] {$\mathtt{a}$} 
				node[below] {$|\mathtt{a}|=\ell^\prime$}(B3);
\draw[edge]
    (B3) -- node[right] {$\mathtt{b}$} (B4);

\node[left=0.1cm of $(B3)!0.5!(B4)$,
      text width=2cm,
      align=right]
{
\small{parallelize over
$\mathtt{b}\in\binom{\ell^\prime}{\ell^\prime-\mathtt{d}}$
}};
\draw[edge] (B4) -- node[below] {$\mathrm{T}$} 
				node[above]{\scalebox{0.5}{$\dim(\mathrm{T})=\mathtt{d}$}}(B5);
\draw[edge] (B5) -- node[left] {$\mathtt{b}\cup B$} (B2);

\draw[edge]
    (B1) to[bend left=20]
    node[above] {$\mathcal{E}$}
    node[below] {$\mathrm{P},\mathrm{Q}$}
    (B2);

\draw[edge]
    (B2) to[bend left=20]
    node[below] {$\mathrm{m}$}
    (B1);
\end{tikzpicture}
}
\caption{\small
Schematic diagram of the algorithm. The index set $\mathtt{a}$ of size $\ell^\prime=\ell/2$ selects a subset of columns from the dense part of $\mathcal{K}$ to form $\mathcal{K}^\prime$. Then the index set $\mathtt{b}$ of size $\ell^\prime - \mathtt{d}$ selects a subset of columns from the dense part of $\mathcal{K}^\prime$ to form the tall matrix $\textsc{M}$. The subspace $\mathrm{T}$ is the left-kernel of $\textsc{M}$, which gives rise to the signature-matrix and then a duplicate in the hashtable yields $B$. 
One can parallelize this algorithm by choosing chunks of $\mathtt{b}$ and sending them to different processors for multiprocessing.
}
\end{figure}
In implementing this algorithm, we used the multiprocessing library in Python. We created all combinations of size 
$\ell^\prime -\mathtt{d}$ from the set $[\ell^\prime]$. Let $\mathtt{b}$ be one such combination.
There are $\binom{\ell^\prime}{\mathtt{d}}$ such subsets, which is a polynomial in $\ell^\prime$ and thus polynomial in $\ell$ for a fixed $\mathtt{d}$. Using the probability table Table~\ref{probtable}, one needs to choose the defect $\mathtt{d}$ so that there is a significant probability of success.
Then for each $\mathtt{b}$, we extract the submatrix $\textsc{M}$ and compute its left-kernel $\mathrm{T}$. We then compute the signature of $\mathrm{T}$ with all the remaining hyperplanes of $\mathcal{K}^\prime$ that are not in $\textsc{M}$ and create the signature-matrix. The size of the signature-matrix is $(\ell^\prime +\mathtt{d})\times{\mathtt{d}}$. Then to compute the hashtable, we compute $\binom{\ell^\prime+\mathtt{d}}{\mathtt{d-1}}$ kernels of matrices of size $(\mathtt{d}-1)\times\mathtt{d}$ which can be a large number. There are ways to speed things up. Use Gray codes to create combinations, where the combinations are ordered, and consecutive ones differ in only one place. Then use  rank-one replacement to compute those kernels. This is faster than computing those kernels independently. We can do a collision-search while creating the hashtable. In this case, after computing a hash, we check if that is already in the hashtable. If not, we put it in. Otherwise, the problem is solved and we return the repetitions.

The readers may find it confusing that we are switching from $\binom{\ell^\prime+\mathtt{d}}{\mathtt{d}-1}$ to $\binom{\ell^\prime+\mathtt{d}}{\mathtt{d}}$ frequently. The rule of thumb is, when we are dealing with the number of kernel computations in the signature-matrix it is the former and the later relates to probability computation.

\begin{table}[htbp]
    \centering
    \caption{Parameter choices with success probability greater than $0.9$. For each repeated value of $\lfloor\log_2 p\rfloor$, the row minimizing $\left\lfloor\log_2 \binom{\ell^\prime + \mathtt{d}}{\mathtt{d} - 1}\right\rfloor+\left\lfloor\log_2 \binom{\ell^\prime}{\mathtt{d}}\right\rfloor$ is shown. The data shows that our attack is no better than exhaustive search.}
    \label{prob-prod}
    \scriptsize
    \setlength{\tabcolsep}{3pt}
    \begin{tabular}{ccccc||@{\qquad}ccccc}
        \toprule
        $\lfloor\log_2 p\rfloor$ & $\mathtt{d}$ & $\left\lfloor\log_2 \binom{\ell^\prime + \mathtt{d}}{\mathtt{d} - 1}\right\rfloor$ & $\left\lfloor\log_2 \binom{l^\prime}{d}\right\rfloor$ & Prob. &
        $\lfloor\log_2 p\rfloor$ & $\mathtt{d}$ & $\left\lfloor\log_2 \binom{\ell^\prime + \mathtt{d}} {\mathtt{d} - 1}\right\rfloor$ & $\left\lfloor\log_2 \binom{\ell^\prime}{\mathtt{d}}\right\rfloor$ & Prob. \\
        \midrule
        100 & 17 & 53 & 48 & 0.98464 & 112 & 19 & 60 & 53 & 0.98205 \\
        101 & 18 & 55 & 49 & 1.00000 & 113 & 20 & 62 & 55 & 1.00000 \\
        102 & 18 & 55 & 49 & 0.99999 & 114 & 20 & 62 & 55 & 0.99998 \\
        103 & 18 & 55 & 49 & 0.99672 & 115 & 19 & 62 & 56 & 1.00000 \\
        104 & 18 & 55 & 49 & 0.94277 & 116 & 19 & 62 & 56 & 0.99876 \\
        105 & 17 & 55 & 51 & 0.95952 & 117 & 19 & 62 & 56 & 0.96485 \\
        106 & 18 & 57 & 52 & 1.00000 & 118 & 20 & 64 & 58 & 1.00000 \\
        107 & 18 & 57 & 52 & 0.99998 & 119 & 20 & 64 & 58 & 0.99998 \\
        108 & 18 & 57 & 52 & 0.99545 & 120 & 20 & 64 & 58 & 0.99535 \\
        109 & 18 & 57 & 52 & 0.93253 & 121 & 20 & 64 & 58 & 0.93178 \\
        110 & 19 & 60 & 53 & 1.00000 & 125 & 20 & 66 & 60 & 0.98227 \\
        111 & 19 & 60 & 53 & 0.99968 &     &    &    &    &         \\
        \bottomrule
    \end{tabular}
\end{table}

\subsection{Pseudocode for the attack algorithm} 
We now present the attack algorithm in pseudocode. The algorithm has two parts. The first one reduces the size of the input matrix by half using an array $\mathtt{a}$, and we call it the recursive step. The next step is to find a zero minor in the reduced matrix using central hyperplane arrangements. First, note that the complexity depends on the defect $\mathtt{d}$ which is fixed by Table~\ref{probtable}. This gives us the expected number of choices of $\mathtt{a}$ required for the success of our algorithm. 
\begin{algorithm}[H]
\small
\caption{Pseudo-code of the algorithm to find a zero minor in $\mathcal{K}$}
\label{alg1}
\begin{algorithmic}[1]
\State Input: a matrix $\mathcal{K}$ of size $\ell\times 2\ell$ and defect $\mathtt{d}$ much smaller than $\ell$. 
\LComment{Using finding a mate.}
	\State{$nr\gets\mathcal{K}.\mathrm{nrows}()$, $nc\gets\mathcal{K}.\text{ncols()}$}
	\State $\mathtt{a} \gets\text{array of length}\; nr/2$ from the set $[nr]$ and $a^\prime$ is its complement from the same set.
	\State $\mathcal{K}^\prime\gets\textrm{submatrix}\left(\mathcal{K}[\mathtt{a}]\right)$
	\Comment {$\mathcal{K}^\prime$ is a submatrix of $\mathcal{K}$ of columns with index from $a$.}
	\State $\mathcal{K}^\prime\gets\mathcal{K}^\prime.\mathrm{transpose}()$
	\State $\mathcal{K}^\prime\gets\textsc{reduce-anti-diagonal}\left(\mathcal{K}^\prime\right)$
	\Comment{Reduces the matrix to anti-diagonal format.}
	\State $\mathcal{K}\gets\mathcal{K}^\prime$.
\LComment{Using intersection poset of hyperplane arrangement.} 
	\State{$nr\gets\mathcal{K}.\mathrm{nrows}()$, $nc\gets\mathcal{K}.\text{ncols()}$}
	\State{$\text{comb}\gets\textsc{all-combinations}$ of $[nr]$ of size $nr - \mathtt{d}$}
\For{$\mathtt{b}$ in comb}
	\Comment{This step can be run in parallel, with each processor getting different $\mathtt{b}$}
	\State{$\textsc{mat}\gets\mathcal{K}[\mathtt{b}]$}; 
	\Comment{$\textsc{mat}$ is the submatrix of $\mathcal{K}$ with columns from $\mathtt{b}$}
	\State {$\mathrm{T}\gets\textsc{left-kernel(mat)}$};
	\Comment{$\dim(\mathrm{T})=\mathtt{d}$.}
	\State $\textsc{signature-mat}\gets\text{mat}(0,\mathtt{d})$
	\Comment{Initialize an empty matrix of zero row and $\mathtt{d}$ cols}
	\For{$i\in b^\prime$} 
	\Comment{$b^\prime=[nr]\smallsetminus{\mathtt{b}}$}
	\State $u\gets$\Call{create-signature}{$\mathrm{H}[i],\mathrm{T}$}
	\Comment{Signature of $\mathrm{T}$ with hyperplane of $i\texttt{th}$ column of $\mathcal{K}$}
	\State $\textsc{signature-mat} = \textsc{signature-mat}$ stack $u$
	\Comment{Adds $u$ at the bottom. }
	\EndFor
	\EndFor
	\State $\mathrm{A}\gets\textsc{signature-mat}\;\text{stack}\; \left(\mathrm{T}.\text{transpose()}\right)$	
	\State $\mathrm{A}^\prime\gets\textsc{create-intersection-signature}(\textrm{A},d)$
	\Comment{$\mathrm{A}^\prime$ is an array of hash-values.}
	\State search for repetition in $\mathrm{A}^\prime$
	\State Return {index of repetitions in $\mathrm{A}^\prime$}

\Procedure{create-signature}{$\mathrm{H},\mathrm{T}$}
\LComment{This creates the signature for the hyperplane $\mathrm{H}$ and the kernel $\mathrm{T}$}
\LComment{Both $H$ and $T$ are basis-matrix of size $(\ell - 1)\times \ell$ and $d\times \ell$ respectively.}
	\State Create a matrix $\mathrm{U}$ by stacking $\mathrm{H}.\text{transpose()}$ and $\mathrm{T}.\text{transpose()}$ horizontally. 
	\State $\mathrm{U}^\prime$ is the RREF of $\mathrm{U}$.
	\State Extract $u$, the last row of $\mathrm{U}^\prime$. 
	\State The vector $u^\prime$ is the last $\mathtt{d}$ elements in the row $u$.
	\State \Return $u^\prime$.
	\Comment{$u^\prime$ is a $1\times\mathtt{d}$ vector.}
\EndProcedure
\Procedure{create-intersection-signature}{$\mathrm{A}, \mathtt{d}$}
	\LComment{Creates signature for penultimate intersections.}
	\State $r\gets \mathrm{H}.\textrm{nrows}()$ and $d\gets\mathrm{H}.\textrm{ncols}()$
	\Comment{Number of rows and columns of $\mathrm{A}$.}
	\State $\texttt{Comb}\gets$ all possible combinations of $\{1,2,\dots,r\}$ of size $\mathtt{d}-1$.
	\State $\mathcal{L} =\textrm{dict}()$ 
	\Comment{Initialize $L$ as an empty dictionary.}	
	\For{$x\in \texttt{Comb}$}
		\State $\mathrm{mat}\gets \mathrm{A}[x]$
		\Comment{mat is the submatrix of $\mathrm{A}$ formed by the rows with index from $x$}
		\State $\mathrm{ker}\gets\textbf{right-kernel}(\mathrm{mat})$
		\Comment{$\mathrm{ker}$ is $1$ dimensional subspace}
		\State $\mathcal{L}[x]\gets\mathrm{hash}(v)$ 
		\Comment {Here $v$ is the generator with first non-zero element 1}
	\EndFor
	\State \Return $\mathcal{L}$
\EndProcedure
\end{algorithmic}
\end{algorithm}
The post-processing is performed once a repetition is found and is fairly easy to describe. There are two kinds of searches: one that stops at the first repetition, whereas the other finds all repetitions. We adopted the first kind in our experiments.  Though it is rare, but still a possibility that two different zero minors have the same central element. We have not observed this event in our computations.
\begin{algorithm}[ht]
\small
\caption{Post Processing}
\label{algo2}
\begin{algorithmic}[1]
\State Input: List of two combinations of size $\mathtt{d} -1$ written as set, $\mathtt{a}$ and $\mathtt{b}$.
\State {$nr\gets\mathcal{K}^\prime.nrows(), nc\gets\mathcal{K}^\prime.ncols()$}
\State {$b^\prime\gets[nr]\smallsetminus(b)$}\Comment{$b^\prime$ is the complement of $b$}
\State Take the union of two sets from input.
\If{cardinality of the union $=\mathtt{d}$}
\State{$\Gamma\gets$ union};\Comment{$\Gamma$ is an array of size $\mathtt{d}$}
\LComment{Now compute the indices in $\mathcal{K}^\prime$ corresponding to the union}
\State{$\textsc{B}=\emptyset$}\Comment{Create an empty list}
\For{$i\in\Gamma$}
	\If{$i<=\mathtt{d}$}
		\State{$\textsc{B}$.append$(b^\prime[i])$}
	\ElsIf{$i>\mathtt{d}$}
		\State{$\textsc{B}$.append$(nc - (i-\mathtt{d})+1)$}
	\EndIf
\EndFor
\State{$\mathtt{b}=\mathtt{b}\cup \textsc{B}$};\Comment{$\mathtt{b}$ are indices for a zero-minor in $\mathcal{K}^\prime$}
\LComment{Now compute indices for zero-minor in the input-matrix $\mathcal{K}$}
	\State{$\textsc{O}\gets[nc,nc+1,\ldots,2nc]$}\Comment{Initialize $\textsc{O}$}
	\For{$i\in \mathtt{b}$}
	\State{$\textsc{O}$.delete$(2\ast nc -i +1)$ for $i \in \texttt{b}$}
	\EndFor
	\State {$\mathtt{a}\gets a\cup\textsc{O}$}
	\State \Return{$\mathtt{a}$}
\LComment{$\mathtt{a}$ are indices of the columns of $\mathcal{K}$ that is a zero-minor.}
	
\ElsIf{cardinality $\neq\mathtt{d}$}
	\State\Return {No result}	
\EndIf
\end{algorithmic}
\end{algorithm}
There is a slight abuse in our introduction of $\mathcal{K}^\prime$ in the Algorithm~\ref{algo2}. The intended meaning is clear.
\subsection{Why does this algorithm work}
As we know from Theorem~\ref{thm1}, one way to find a zero minor in $\mathcal{K}$ is to find $\ell$ columns in $\mathcal{K}$ whose hyperplanes form a central hyperplane arrangement. Equivalently, find an array $\Upsilon$ of length $\ell$ from the set $[2\ell]$, such that, 
$\cap_{i\in\Upsilon}\mathcal{K}_i\neq 0$. Then there is a non-zero $v$, such that, $v\in\cap_{i\in\Upsilon}\mathcal{K}_i$. The algorithm that we developed is finding this $v$. Now clearly $v\in\cap_{i\in\Upsilon}\mathcal{K}_i$ is equivalent to the condition that $v\in\cap_{i\in\Upsilon^\prime}\mathcal{K}_i$, for every $\Upsilon^\prime$ a subset of $\Upsilon$ of size $\ell-1$. Then, an obvious thing to do is to take all $\ell-1$ subsets of the set $[2\ell]$ and check the intersection of the corresponding hyperplanes for repetitions. We focus on the left-half and the right-half of the rectangular matrix $\mathcal{K}$ simultaneously. This way, we divide the problem into two. 

We choose a combination $\mathtt{b}$ of size $\ell - \mathtt{d}$, where $\mathtt{d}$ is small. We call $\mathtt{d}$ the defect in $\mathtt{b}$. We define $\mathrm{T} = \cap_{i\in\mathtt{b}}\mathcal{K}_i$, where $\mathcal{K}_i$ is the hyperplane corresponding to the $i\in\mathtt{b}$. Assuming that the hyperplanes are in general position, the dimension of $\mathrm{T}$ is $\mathtt{d}$.

There is some similarity of our algorithm with the well known baby-step giant-step algorithm. The outer loop of our algorithm which loops over all possible choices of $\mathtt{b}$ is the giant-step and the inner loop that loops over all possible submatrices of size $\mathtt{d}-1$ in the signature-matrix is the baby step. However, we could not find a way to define the collision that reveals a central hyperplane intersection. We find it an interesting question, to examine this algorithm in the light of a collision search algorithm. Then the complexity changes drastically because of the birthday-paradox.

Now we need to find $\mathtt{d}$ hyperplanes from $\Delta=[2\ell]\smallsetminus(\mathtt{b})$, such that, the intersection of these hyperplanes with $\mathrm{T}$ is non-zero. This is the same as finding two different sets of size $\mathtt{d}-1$ from $\Delta$, such that the nonzero vector $v$ in the intersection of each set of $\mathtt{d} -1$ hyperplanes along with $\mathrm{T}$ is the same. Size of $\Delta$ is $\ell + \mathtt{d}$. If $\mathtt{d}$ is fixed, the number of choices ${\ell + \mathtt{d}}\choose{\mathtt{d}-1}$ is polynomial in $\ell$. So the total number of choices is ${{\ell}\choose{\mathtt{d}}}{{\ell + \mathtt{d}}\choose{\mathtt{d}-1}}$ which is polynomial in $\ell$. Now, for a fixed choice of $\mathtt{b}$, take all $\mathtt{d}-1$ choices from $\Delta$, and compute the signature using Theorem~\ref{thm2} for the penultimate intersections. The above argument works for $\mathcal{K}$ as well as $\mathcal{K}^\prime$. The above argument will be formalized below in the form of a theorem. 
\begin{theorem}
Let $\Upsilon, \Delta, \mathrm{T}$ and $\mathtt{d}$ be as defined earlier. Let $x$ be a $\mathtt{d}-1$ subset of $\Delta$. Then for every $i\in x$ there is a signature $\mathtt{s}_i$ as computed in Theorem~\ref{thm2}. Make a matrix $X$ with these signatures as rows. Thus the size of $X$ is 
$(\mathtt{d}-1)\times \mathtt{d}$. Assuming general position, the right-kernel of $X$ is one dimensional. Assume that it is generated by $v$. Then $\cap_{i\in \Upsilon\cup x}\mathcal{K}_i=\langle v\mathrm{T}\rangle$. Furthermore, by making the first non-zero element of $v$ one, $v$ is unique up to a scalar multiple. 
\end{theorem}
\begin{proof}
Recall that $|\Upsilon| = \ell$. The principal argument is, assuming general position, the intersection of $t$ hyperplanes is a subspace of dimension $\ell-t$, where $0<t<\ell$. For a non-zero minor, the dimension of the intersection of $\ell-1$ hyperplanes is $1$ and the intersection of $\ell$ hyperplanes is $0$. For a zero-minor, the dimension of the intersection of $\ell$ hyperplanes, its ultimate intersection, is also one. When that happens, all the penultimate intersections are the same.

With this understanding, look at $X$ and its right-kernel. The right-kernel being penultimate intersection is one-dimensional $\langle v\rangle$. And, the same $v$ for two different $x$ leads to a one-dimensional ultimate intersection. The indices of these two different $x$ gives rise to a zero-minor.
\end{proof}
Now once we have a vector $v$ for every subset of size $\mathtt{d}-1$ of $\Delta$ we have an element $v\mathrm{T}\in\cap_{\Upsilon\cup x}\mathcal{K}_i$, where $x$ a subset of $\Delta$ of size $\mathtt{d}-1$. We then compute and store the hash of this \emph{representative} vector $v$ along with $x$; for all possible $x$ and call it the hashtable. If there is a repetition in the hash-table for $x$ and $x^\prime$, we can construct a set of size $\ell$ from these sets $x$, $x^\prime$ and $\Upsilon$ and the corresponding minor of $\mathcal{K}^\prime$ will be zero. This proves the following theorem:
\begin{theorem}
Let $\mathcal{K}$ be a matrix of size $\ell\times 2\ell$. If any two penultimate intersections of hyperplanes from the columns of $\mathcal{K}$ are the same as subspaces, then there is a zero minor in $\mathcal{K}$. The indices of that zero minor can be determined from the union of the indices of these two penultimate intersections which have the same representative nonzero vector.   
\end{theorem}
We repeat, the representative vector $v$ of the penultimate intersection must be normalized by setting its first nonzero entry equal to $1$ using scalar multiplication. There is also a remote possibility that two different ultimate intersections are represented by the same nonzero vector. But if that happens, it will be a rare event inside a rare event. We ignored it.
\section{The probability of success}
We recall from our earlier work~\cite[Theorems 4 \& 5]{first} that the probability of success of the Las Vegas algorithm is greater than $0.5$ as the size of the group of prime order $\mathtt{p}$ tends to infinity. The estimate was elementary and depended on computing the number of unique partitions of $\mathrm{m}$ modulo $\mathtt{p}$. In particular we computed, given a positive integer $\mathrm{m}$, the number of partitions of $\mathrm{m}$ into $k$ unique parts in the range $(1,\mathtt{p})$ is given by the following expression:
\begin{equation}
\dfrac{(\mathtt{p}-1)(\mathtt{p}-2)\cdots(\mathtt{p}-k+2)(\mathtt{p}-2k+1)}{k!}
\end{equation}
from which the probability of success can be computed as bounded below by 
\begin{equation}
\dfrac{\mathtt{p}-2k+1}{\mathtt{p}(\mathtt{p}-k+1)}
\end{equation}
which is really close to $1/\mathtt{p}$.

The idea behind computing this expression is easy. We try to fill $k$ holes with integers from $(1,\mathtt{p})$ such that they add up to $\mathrm{m}$. This can be done in 
$(\mathtt{p}-1)(\mathtt{p}-2)\cdots(\mathtt{p}-k-1)/k!$ ways, when order does not matter. Given that the total number of choices are 
$\binom{\mathtt{p}}{k}$, the probability is $1/\mathtt{p}$. The situation becomes slightly more complicated when we have to consider the fact that repetitions are not allowed. But, since $\mathtt{p}$ is much larger than $\ell$ or $\ell^\prime$, this is a pathology that can be ignored in practice. However, we did not ignore it in our work. 
 
Now recall that the points that are in the intersection of the elliptic curve and the complete curve $\mathcal{C}$, is a subset of those points that created the matrix $\mathcal{M}$ and are the ones that are not in the column indices of the zero minor~\cite[Corollary 1]{first}. 

For our current algorithm, we are looking at a conditional probability: namely, the probability of success when we have chosen an array $\mathtt{a}$ of size $\ell^\prime$ to be part of a zero minor? Let $\mathtt{a}^\prime = [\ell]\smallsetminus(\mathtt{a})$. Moreover, let $\mathtt{t}$ be the sum of the $n_i$ where $i\in \mathtt{a}^\prime$ modulo $\mathtt{p}$.
Once $\mathtt{a}$ is chosen from the dense part, we perform the recursive reduction corresponding to it. The output is a matrix $\mathcal{K}^\prime$ of size $\ell^\prime\times\ell$ in anti-diagonal format, and we are looking for a zero minor in it using central hyperplane arrangements. Now from Equation~\ref{eqn1}, we are looking at $\ell^\prime$ distinct integers $-n_j$, such that $\mathrm{m}$ times their sum is congruent to $\mathtt{t}\mod(\mathtt{p})$. 

In the determine part of our attack, we are working with the transpose of the matrix $\mathcal{K}^\prime=\mathcal{K}[\mathtt{a}]$, which is of size $\ell^\prime\times\ell$. We fix $\ell^\prime -\mathtt{d}$ columns from the dense part of $\mathcal{K}^\prime$ and then try to find $\mathtt{d}$ columns from the rest of the columns that extend the fixed part to a zero minor in $\mathcal{K}^\prime$. There are $\binom{\ell^\prime +\mathtt{d}}{\mathtt{d}}$ such choices for this fixed part. The probability that none of these choices sum up to the required sum of $\mathtt{t}:=\mathtt{t}/\mathrm{m} \mod{\mathtt{p}}$ is 
\begin{equation}\label{defect_eqn}
\Lambda=\left(1-\dfrac{\mathtt{p}-2\mathtt{d} +1}{\mathtt{p}(\mathtt{p}-\mathtt{d} +1)}\right)^{\binom{\ell^\prime+\mathtt{d}}{\mathtt{d}}}
\end{equation}    
Furthermore, this has to happen for all choices $\mathtt{b}$ from $[\ell]$. The number of such choices are 
$\binom{\ell^\prime}{\mathtt{d}}$. Thus none of the choices of $\mathtt{b}$ add upto $\mathtt{t}$ is $\Lambda^{\binom{\ell^\prime}{\mathtt{d}}}$.
Thus, the probability that at least one of these choices will add up to $\mathtt{t}$ is $1-\Lambda^{\binom{\ell^\prime}{\mathtt{d}}}$. We computed the probability for primes of size upto $2^{550}$ and defects from $15$ to $50$. However, presenting such a large table is difficult. So in Table~\ref{probtable}, we present the first $\mathtt{d}$ which achieves a significant probability and the primes are presented at a gap of $10$. 
\begin{table}[htbp]
  \centering
  \small
  \setlength{\tabcolsep}{4pt}
  \caption{Minimum defect required to attain the indicated success probability. A blank entry means that the target probability was not attained for any tested defect $15\leq \mathtt{d}\leq 50$.}
  \label{probtable}
  \scriptsize
  \begin{tabular}{r|cccc||@{\qquad}r|cccc}
    \hline
    $\lfloor\log_2 \mathtt{p}\rfloor$ & $0.25$ & $0.5$ & $0.75$ & $0.99$ &
    $\lfloor\log_2 \mathtt{p}\rfloor$ & $0.25$ & $0.5$ & $0.75$ & $0.99$ \\
    \hline
    150 & 15 & 15 & 15 & 15 & 360 & 34 & 34 & 34 & 35 \\
    160 & 16 & 16 & 16 & 16 & 370 & 35 & 35 & 35 & 35 \\
    170 & 17 & 17 & 17 & 17 & 380 & 36 & 36 & 36 & 36 \\
    180 & 18 & 18 & 18 & 18 & 390 & 37 & 37 & 37 & 37 \\
    190 & 18 & 19 & 19 & 19 & 400 & 38 & 38 & 38 & 38 \\
    200 & 19 & 20 & 20 & 20 & 410 & 39 & 39 & 39 & 39 \\
    210 & 20 & 20 & 21 & 21 & 420 & 40 & 40 & 40 & 40 \\
    220 & 21 & 21 & 21 & 22 & 430 & 40 & 41 & 41 & 41 \\
    230 & 22 & 22 & 22 & 23 & 440 & 41 & 42 & 42 & 42 \\
    240 & 23 & 23 & 23 & 24 & 450 & 42 & 42 & 43 & 43 \\
    250 & 24 & 24 & 24 & 24 & 460 & 43 & 43 & 43 & 44 \\
    260 & 25 & 25 & 25 & 25 & 470 & 44 & 44 & 44 & 45 \\
    270 & 26 & 26 & 26 & 26 & 480 & 45 & 45 & 45 & 46 \\
    280 & 27 & 27 & 27 & 27 & 490 & 46 & 46 & 46 & 46 \\
    290 & 28 & 28 & 28 & 28 & 500 & 47 & 47 & 47 & 47 \\
    300 & 29 & 29 & 29 & 29 & 510 & 48 & 48 & 48 & 48 \\
    310 & 29 & 30 & 30 & 30 & 520 & 49 & 49 & 49 & 49 \\
    320 & 30 & 31 & 31 & 31 & 530 & 50 & 50 & 50 & 50 \\
    330 & 31 & 31 & 32 & 32 & 540 &  &  &  &  \\
    340 & 32 & 32 & 32 & 33 & 550 &  &  &  &  \\
    350 & 33 & 33 & 33 & 34 &  &  &  &  &  \\
    \hline
  \end{tabular}
\end{table}

\section{Implementation}
The algorithm was implemented for prime fields using SageMath~\cite{sagemath} and FLINT~\cite{flint} together with Python's multiprocessing library. The main objective of this implementation was a proof of concept for the algorithm. The implementation went well. However, at this stage, we do not have much data to share. 

The algorithm has two parts. The input is a matrix $\mathcal{K}$ of size $\ell\times 2\ell$ in the anti-diagonal format. We assume that $\ell$ is even and $\ell = 2\ell^\prime$. Then we choose a random vector $\mathtt{a}$ of size $\ell^\prime$ from the set $[\ell]$ and an positive integer $\mathtt{d}$, called the defect. Then we compute $\mathcal{K}^\prime$ from $\mathcal{K}$ as the submatrix of columns whose indices are in $\mathtt{a}$. Then we transpose this $\mathcal{K}^\prime$ and reduce it to anti-diagonal format using only row operations. Then $\mathcal{K}^\prime$ is a matrix of size $\ell^\prime\times\ell$. This part was done serially. 

Now the parallel step begins. We make a list of all possible combinations of $\left[\ell^\prime\right]$ of size $\ell^\prime-\mathtt{d}$. Let $\mathtt{b}$ denote such a combination, and each processor got a chunk of $\mathtt{b}$. For each assigned $\mathtt{b}$, a processor computes the submatrix $\textsc{M}$ with column indices in $\mathtt{b}$ from $\mathcal{K}^\prime$ and then computes its left-kernel $\mathrm{T}$. The dimension of $\mathrm{T}$ is $\mathtt{d}$. Then corresponding to this fixed $\mathrm{T}$, we compute the signature-matrix and look for duplicates. If any one of these processors finds a duplicate, it returns both entries in the duplicate, and the computation moves into the post-processing stage, where the discrete logarithm is computed. 

Because of our ability to use the structure of the sparse part of the matrix which is in anti-diagonal format, computing the signature-matrix was fast. We computed RREF for only $\mathtt{d}$ rows in the signature-matrix, and the rest were free. Thus computing the signature-matrix was efficient. One bottleneck of the computation was computing maximal-minors in this signature-matrix. Though these were small $\mathtt{d}\times\mathtt{d}$ matrices and $\mathtt{d}$ was a small integer and the number of such minors was polynomial in $\ell+\mathtt{d}$. But still, there were a lot of them. Furthermore, we were not able to use parallel computing on these minors. This is because it will be a two-level parallelization. We did not have the expertise or the resources for that. 

In a \emph{practical attack} on the discrete logarithm problem, we will have the prime $\mathtt{p}$. From $\mathtt{p}$ we will compute $\ell$. Then the first step would be to determine the defect $\mathtt{d}$ which gives us a significant probability of success from Equation~\ref{defect_eqn}. Then we use that $\mathtt{d}$ in our algorithm. In that case, the computational resources needed can be easily determined, and the algorithm will run in time polynomial in $\ell$. However, the degree of that polynomial will be big.  
\section{Conclusion}
This paper is a step forward in the Las Vegas attack that we developed about a decade ago. The work presented here is both theoretical and experimental. The implementation done was mostly a proof of concept. The implementation went well and the algorithm presented here was rigorously tested and passed all tests.

Ultimately, this paper reduces solving the discrete logarithm problem to searching for duplicates in an array, which we call the hashtable. The length of this hashtable is $\binom{\ell^\prime + \mathtt{d}}{\mathtt{d}}$. The hashtable is created from hashes of vectors. Each vector is the generator of the kernel of $\mathtt{d}-1$ rows of the signature-matrix. The size of the signature-matrix is $(\ell^\prime+\mathtt{d})\times \mathtt{d}$. The kernel is one-dimensional, and the hash of its generator is the hash value in the hashtable. We created the signature from a single fixed $\mathtt{b}$, an array of size $\ell^\prime-\mathtt{d}$ from the set $\left[\ell^\prime\right]$. The matrix $\textsc{M}=\mathcal{K}^\prime[\mathtt{b}]$ is extracted and then $\mathrm{T}$ is the left-kernel of $\textsc{M}$. Thus, fixing $\mathtt{b}$ determines $\textsc{M}$ and its left-kernel $\mathrm{T}$ which is of dimension $\mathtt{d}$. Corresponding to this $\mathrm{T}$, we compute signatures from all hyperplanes coming from the columns of $\mathcal{K}^\prime$ whose indices are not in $\mathtt{b}$ and the signature-matrix is formed with rows as signatures. 
Then the $\mathtt{d}-1$ rows, whose right-kernel is one-dimensional, represents the penultimate intersection of $\ell^\prime -1$ hyperplanes. The hashtable is formed with the hash of the representatives on this one-dimensional subspaces. 
If there is a repetition in the hashtable, the union of these two sets of indices which gave rise to those hashes that collided is a non-zero ultimate intersection. Thus $\mathtt{b}$ union the union of these two sets is the set of indices of columns of $\mathcal{K}^\prime$ that form a zero minor. This set is in one-one correspondence with a set of indices from the sparse part of $\mathcal{K}$ and adjoining these to $\mathtt{a}$ is the indices of the columns of a zero minor in $\mathcal{K}$ and solves the discrete logarithm problem.

Now, to the question of the algorithm's success probability. At the end, the probability estimate is roughly
$1 - \left(1 - \dfrac{1}{\mathtt{p}}\right)^{\binom{\ell^\prime+\mathtt{d}}{\mathtt{d}}\binom{\ell^\prime}{\mathtt{d}}}$. Since $1-\tfrac{1}{\mathtt{p}}$ approaches $1$ as $\mathtt{p}$ becomes large, the product $\binom{\ell^\prime+\mathtt{d}}{\mathtt{d}}\binom{\ell^\prime}{\mathtt{d}}$ has to be large enough for a reasonable chance of success. This is what determines the choice of the number of recursions and the defect $\mathtt{d}$. As we mentioned earlier, there are two bottlenecks in implementing this algorithm. One is the number of choices of $\mathtt{b}$, and the other is the number of choices of $\mathtt{d}-1$ rows in the signature matrix. 

The product of binomial coefficients $\binom{\ell^\prime+d}{d}\binom{\ell^\prime}{\mathtt{d}}$ is special. On the one hand, it is a bottleneck in the algorithm because as it can become very large. On the other hand, while computing the probability, it is the exponent which makes the probability large enough for computations to be worthwhile. A careful balance is therefore required to keep the computations feasible while maintaining a high success probability. This is possible from Equation~\ref{defect_eqn}, which was used to compute Table~\ref{probtable}. 

We look at the signature matrix one last time. To solve the discrete logarithm problem, we need to find a maximal minor in it. This is a theme that keeps recurring: to find a zero minor in a matrix, we have to find a zero minor in another matrix, \emph{albeit} smaller. Recall our construction of $\mathcal{K}$, $\mathcal{K^\prime}$ and then the signature-matrix, and the fact that a zero-minor in the signature-matrix, gives rise to a zero-minor in $\mathcal{K}^\prime$, which in turn gives rise to a zero-minor in $\mathcal{K}$.
Is this phenomenon just a pure coincidence, or is there some structure underneath that remains unexposed?  

\paragraph{Acknowledgements} This work was partially supported by an NBHM research grant and a SERB Matrics grant. The AI agent ChatGPT was used for coding purposes. Ansari Abdullah implemented a preliminary version of the attack algorithm. SageMath and FLINT was used for coding purposes. 

\bibliographystyle{plain}
\renewcommand{\baselinestretch}{0.5}
\bibliography{ref}
\end{document}